\pdfoutput=1

\documentclass[envcountsect]{llncs}
\usepackage[T1]{fontenc}
\usepackage[utf8]{inputenc}
\usepackage{ upgreek }
\usepackage[procnumbered,ruled, vlined, linesnumbered, noresetcount]{algorithm2e}
\usepackage[english]{babel}

\usepackage{cancel} 
\usepackage{enumerate}
\usepackage[mathscr]{euscript}
\usepackage{graphicx} 

\usepackage{lmodern}

\usepackage{paralist}

\usepackage{varioref}

\usepackage{latexsym,amsmath,amssymb,textcomp,amsthm}
\usepackage{listings} 
\usepackage{pgf}
\usepackage{semtrans}
\usepackage{rotating}
\usepackage{subfig}

\usepackage{ stmaryrd }
\usepackage{xspace} 
\newcommand{\macro}[2]{ \providecommand{#1}{{\ensuremath{#2}}\xspace}}

\macro{\algo}{\mathcal{E}}
\macro{\A}{\mathcal{A}}
\macro{\cP}{\mathcal{P}}
\macro{\area}{\mathcal{A}}

\macro{\bigdot}{\bullet}
\macro{\boul}{B_{\triangles}}
\macro{\by}{\mathcal{Y}}

\macro{\CCC}{\mathcal{CIR}}
\macro{\Cn}{\mathcal{C}_{\nC}}
\macro{\cate}{\textsc{SameDir}}

\macro{\cir}{\texttt{Cluster}}
\macro{\comp}{\texttt{t}}

\macro{\dup}{dup}
\newcommand{\dest}[3]{\textsc{dest}_{#1}(#2,\lambda(#3))}
\newcommand{\destc}{\textsc{dest}}
\newcommand{\destjc}[3]{\textsc{dest}_{#1}(#2,#3)}

\macro{\deltaM}{\delta^{\mapG}}

\macro{\UpdateStructures}{\texttt{updateStructures}}
\macro{\ProcTransitiveClosure}{\texttt{computeTransitionalClosure}}
\macro{\UpdateMap}{\texttt{updateMap}}
\macro{\IdentifyComponents}{\texttt{identifyNewComponents}}

\macro{\equi}{\textsc{V-edges}}

\macro{\eqcl}{\mathcal{C}}
\macro{\equiB}{\,\--_B\,}

\usepackage{xfrac}    

\macro{\Gu}{\hat{G}}
\macro{\Ku}{\hat{K}}

\macro{\qcc}{\Phi}

\macro{\NT}{\mathcal{NT}}
\macro{\newcomp}{\mathcal{C}}
\macro{\Balls}{\mathcal{B}}
\macro{\Neig}{N^{\triangle}}

\macro{\FNT}{\mathcal{FC}}
\macro{\F}{\mathcal{F}}

\macro{\gfam}{\mathcal F}

\macro{\G}{\mathcal{G}}
\macro{\getBino}{\texttt{getBino}}
\macro{\Gport}{\mathcal{G}^{\delta}}
\macro{\Gu}{\hat{G}}

\macro{\hmax}{h_{max}}
\macro{\Cmax}{V(c)^{\hmax}}
\macro{\Cpmax}{V(c')^{\hmax}}

\macro{\cyclecomp}{\mathcal{C}}

\macro{\INF}{\mathcal{IC}}
\macro{\IntCond}{\textsc{int}}

\macro{\labels}{labels}
\macro{\lab}{\lambda}
\macro{\listC}{\textsc{ListComp}}

\newcommand{\lnivsup}[2]{
  \Big(#1,\big(#2\big)\Big)
}

 \macro{\mapCCC}{\cir}
\macro{\mapG}{\textsc{Map}}

\macro{\seen}{\textit{seen}}

\macro{\expl}{\textit{expl}}

\macro{\namemaps}{\labels}

\macro{\n}{\texttt{n}}
\macro{\w}{\texttt{w}}
\macro{\m}{\texttt{m}}
\macro{\nC}{\texttt{n}_{\texttt{C}}}

\macro{\seen}{{Seen}}
\macro{\nV}{\texttt{n}_{\texttt{V}}}
\macro{\p}{\texttt{p}}
\macro{\q}{\texttt{q}}
\macro{\N}{\mathbb{N}}
\macro{\cluster}{\texttt{C}} 
\macro{\idcomp}{\texttt{n}_\texttt{C}} 
\macro{\intersect}{\texttt{intersect}} 
\macro{\neigh}{\texttt{neigh}}
\macro{\borderM}{\partial \mapG}
\macro{\border}{\texttt{borders}}
\macro{\cloequiv}{\equiv^*}
\macro{\Requiv}{\mathcal{R}_{\equiv}}
\macro{\ncloequiv}{{\not\equiv}^*}
\newcommand{\prevert}{\texttt{PRE-VERT}}
\macro{\quotient}{\sfrac{\prevert}{\cloequiv}}

\newcommand{\newvert}[1]{\overline{[(#1)]}}

\macro{\vis}{\texttt{Vis}}

\macro{\Next}{next}
\macro{\nivsup}{\texttt{Hor}}
\macro{\nspcc}{\mathcal{I}_{ST}}
\macro{\nstV}{\mathcal{N}_{st}}
\macro{\nbfs}{\mathcal{N}_{bfs}}
\macro{\nCT}{\#CompTrav}
\macro{\no}{\textsc{no}}
\macro{\ord}{\lessdot}

\newcommand{\parag}[1]{\vspace{-0.18cm}\paragraph{\textbf{#1}.}}
\macro{\parc}{\textsc{Path}}
\macro{\r}{\texttt{r}}
\macro{\path}{\textsc{Path}}
\macro{\s}{\texttt{s}}

\macro{\pn}{\mathcal{P}}

\newcommand{\port}[2]{\delta_{#1}(#2)}
\macro{\pred}{pred}
\macro{\rayon}{radius}

\macro{\rec}{\mathcal R}
\macro{\ringTF}{\mathcal{R}^{\geq 4}}

\macro{\SC}{\mathcal{SC}}
\macro{\Stree}{\mathcal{ST}}

\macro{\spac}{\mathbb{S}}
\macro{\stac}{\textsc{Stack}}
\macro{\ske}{\mathcal{G}^1}
\macro{\Spe}{\mathcal{S}}

\macro{\TF}{triangle-free}
\newcommand{\tild}[1]{\widetilde{#1}}
\macro{\tu}{\tild{u}}
\macro{\tv}{\tild{v}}
\macro{\tts}{\tild{s}}
\macro{\tx}{\tild{x}}
\macro{\tw}{\tild{w}}

\macro{\Topo}{\mathbb{T}}
\macro{\tree}{\mathcal{T}}
\macro{\treeC}{\mathcal{T}_{\mathcal{C}}}
\macro{\TCond}{\textsc{tri}}

\macro{\view}{\mathcal{T}}
\macro{\haut}{\shortuparrow}

\macro{\Xu}{\hat{\X}}
\macro{\X}{\mathcal{K}}
\macro{\XG}{\X(G)}

\macro{\yes}{\textsc{yes}}

\theoremstyle{plain}
\newcounter{cthm}
\newtheorem{thm}[cthm]{Theorem}
\newtheorem{lem}{Lemma}[section]

\newtheorem{rem}[lem]{Remark}
\newtheorem{prop}[lem]{Proposition}
\newtheorem{cor}[lem]{Corollary}
\newtheorem{mydef}[lem]{Definition}

\institute{LIF, Université Aix-Marseille and CNRS, FRANCE}

\title{Using Binoculars for Fast Exploration and Map Construction 
in Chordal Graphs and Extensions 
}
\author{J\'er\'emie Chalopin, Emmanuel Godard and Antoine Naudin
}

\begin{document}

\maketitle
\pagestyle{plain}
\setcounter{page}{1}
\begin{abstract}

We investigate the exploration and mapping of anonymous graphs by a mobile agent.  It is
long known that, without global information about the graph,
it is not possible to make the agent halt after the exploration
except if the graph is a tree.  
We therefore endow the agent
with \emph{binoculars}, a sensing device that can show the local
structure of the environment at a constant distance of the agent's
current location and investigate networks that can be efficiently
explored in this setting.

In the case of trees, the exploration without binoculars is fast
(i.e. using a DFS traversal of the graph, there is a number of moves
linear in the number of nodes).
We consider here the family of Weetman graphs that is a generalization of the 
standard family of chordal graphs and present a new
deterministic algorithm that realizes Exploration of any Weetman graph, without
knowledge of size or diameter and for any port
numbering.
The number of moves is linear in the number of nodes, despite the fact
that Weetman graphs are not sparse, some having a number of edges that is
quadratic in the number of nodes.

At the end of the Exploration, the agent has also computed a map of
the anonymous graph.

{\bf Keywords:} Mobile Agent, Graph Exploration, Map Construction, Anonymous Graphs, 
 Linear Time, Chordal Graphs, Weetman graphs

\medskip
{\bf Eligible for the Best Student Paper Award}
\end{abstract}

\section{Introduction}

Mobile agents are computational units that can progress autonomously
from place to place within an environment, interacting with the
environment at each node that it is located on. These can be hardware
robots moving in a physical world or software robots.
 Such software robots
(sometimes called bots, or agents) are already prevalent in the
Internet, and are used for performing a variety of tasks such as
collecting information or negotiating a business deal.
More generally, when the data is physically
dispersed, it can be sometimes beneficial to move the computation to
the data, instead of moving all the data to the entity performing the
computation. The paradigm of mobile agent computing / distributed
robotics is based on this idea.  
As underlined in \cite{Das_beatcs}, the use of mobile agents 
has been advocated for numerous reasons
such as robustness against network disruptions, improving
the latency and reducing network load, providing more
autonomy and reducing the design complexity, 
and so on (see e.g. \cite{7mobile}).

For many distributed problems with mobile agents, exploring, that is
visiting every location of the whole environment, is an important
prerequisite. In its thorough exposition about Exploration by mobile
agents \cite{Das_beatcs}, S. Das presents numerous variations of the
problem. 
In particular, it can be noted that,
given some global information about the environment
(like its size or a bound on the diameter), it is always possible to
explore, even in environments where there is no
local information that enables to know, arriving on a node, whether it
has already been visited (e.g. anonymous networks).  If no global
information is given to the agent, then the only way to perform a
network traversal is to use an \emph{unlimited} traversal
(e.g. with a classical BFS or Universal Exploration Sequences
\cite{AKLLR79,uxs,R08} with increasing parameters).
This infinite process is sometimes called
\emph{Perpetual Exploration} when the agent visits
infinitely many times every node. 
Perpetual Exploration has application mainly to
security and safety when the mobile agents are a way to regularly check that the
environment is safe.  
But it is important to note that in the case where no global information is
available, it is impossible to always detect when the
Exploration has been completed. This is problematic when one would
like to use the Exploration algorithm composed with another
distributed algorithm.
In this note, we focus on fast Exploration with termination. It is known
that in general anonymous networks, the only topology that enables to
stop after the exploration is the tree-topology.  From standard
covering and lifting techniques, it is possible to see that exploring
with termination a (small) cycle would lead to halt before a complete
exploration in huge cycles.
Moreover, using a simple DFS traversal, Exploration on trees has cover
time that is linear in the number of nodes.

We have shown in \cite{CGN15} that it is possible to 
explore, with full stop, non-tree topologies without global
information using some local information.
The information that is provided can be informally described as giving
\emph{binoculars} to the agent. This constant range sensor enables the
agent to ``see'' the graph (with port
numbers) that is induced by the
adjacent nodes of its current location. 
See Section~\ref{sect:def} for a formal definition.

Using binoculars is a quite natural enhancement
for mobile robots.
In some sense,
we are trading some a priori global information (that might be difficult to
maintain efficiently) for some local information that
the agent can \emph{autonomously} and \emph{dynamically} acquire.

In  \cite{CGN15}, a complete characterization of
which networks can be explored with binoculars is given and the
exploration time is proven to be not practicable for the whole family
of networks that can be explored with binoculars.
Here we focus on families that can be explored in a fast way, 
typically in a time linear in the number of nodes.
  

Chordal graphs are tree-like graphs where ``leafs'' are so-called
simplicial vertices, {\it i.e.} vertices whose neighbourhood is a clique.
Using binoculars, it is possible to locally detect such simplicial
vertices. But how to leverage such detection to get an exploration
algorithm in anonymous graphs is not straightforward since  
it is not possible to mark nodes.

\parag{Our Results}

We present an algorithm that efficiently explores all chordal graphs in a
linear number of moves by a mobile agent using binoculars.
The main contribution is
that the exploration is fast even if the agent does not know the size
or the diameter (or bounds).
The algorithm actually leverages properties of chordal graphs that are verified
in a larger class of graphs: the Weetman graphs. 
This family has been introduced by Weetman \cite{weetman_locally_94}
and can be defined with metric local conditions (see later). 
It contains the family of Johnson graphs.

Using binoculars, we therefore show it is possible to explore and map
with halt dense graphs (having a number of edges quadratic in the number of
nodes) in $O(n)$ moves, for any port numbering. 

 %

\vspace{-0.2cm}
\parag{Related works}
To the best of our knowledge, efficient Exploration 
using binoculars has never been considered for mobile agent on graphs.
When the agent can only see the
label and the degree of its current location, it is well-known that
any Exploration algorithm can only halt on trees and a standard DFS
algorithm enables to explore any tree in $O(n)$ moves. Gasieniec et
al.~ \cite{explotree} presented an algorithm that can explore any tree
with a memory of size $O(\log n)$. For general anonymous graphs,
Exploration with halt has mostly been investigated assuming at least
some global bounds, in the goal of optimizing the move complexity. It
can be done in $O(\Delta^n)$ moves using a DFS traversal while knowing
the size $n$ when the maximum degree is $\Delta$.  This can be reduced
to $O(n^3 \Delta^2 \log n)$ using Universal Exploration
Sequences \cite{AKLLR79,uxs} that are sequences of port numbers that
an agent can sequentially follow and be assured to visit any vertex of
any graph of size at most $n$ and maximum degree at most $\Delta$.
Reingold~\cite{R08} showed that universal exploration sequences can be
constructed in logarithmic space.

Trade-offs between time and memory for exploration of anonymous tree
networks has been presented in \cite{explotree}.  
Note that in this case, the knowledge of the size is required to halt
the exploration.
For example, there is a very simple exploration algorithm for cycles
(``go through the port you are not coming from'') that needs the
knowledge of the size to be able to halt.
Here we are looking for algorithm that does not use (explicitly or
implicitly) such knowledge. 

Trading global knowledge for structural local information by designing
specific port numbering, or specific node labels that enable easy or
fast exploration of anonymous graphs have been proposed in
\cite{CFIKP05,Gasieniec_Radzik_2008,Ilcinkas_2008}.
Note that using binoculars is a local information that can be locally
maintained contrary to the schemes proposed by these papers where the
local labels are dependent of the full graph structure.

See also \cite{Das_beatcs} for a detailed discussion about Exploration using
other mobile agent models (with pebbles for examples).

\section{Exploration with Binoculars}

\subsection{The Model}

\parag{Mobile Agents}
We use a standard model of mobile agents, that we now formally describe.
\pagebreak[3]
A mobile agent is a computational unit evolving in an undirected 
simple graph $G=(V,E)$ from vertex to vertex along the edges. 
A vertex can have some label attached to it. 
There is no global guarantee on the labels, in particular
vertices have no identity (anonymous/homonymous setting),
 i.e., local labels are not guaranteed to be unique.  
The vertices are endowed with a port numbering function available to
the agent in order to let it navigate within the graph. 
Let $v$ be a vertex, we denote by
$\delta_v:V\to \N$,  the injective port numbering function giving a
locally unique identifier to the different adjacent nodes of $v$.
We denote by $\port{v}{w}$ the port number of $v$ leading to the vertex $w$,
i.e., corresponding to the edge $vw\in E(G)$.
We denote by $(G,\delta)$ the graph $G$ endowed with a port numbering $\delta=\{\delta_v\}_{v\in  V(G)}$.

When exploring a network, we would like to achieve it for any port
numbering.
So we consider the set of every graph endowed with a valid port
numbering function, called $\Gport$. 
By abuse of notation, since the port numbering is usually fixed, we
denote by $G$ a graph $(G,\delta)\in\Gport$.

The behaviour of an agent is cyclic: it obtains local information
(local label and port numbers), computes some values, and moves to its
next location according to its previous computation.  We also assume
that the agent can backtrack, that is the agent knows via which port
number it accessed its current location.  
We do not assume that the starting point of the agent (that is
called the \emph{homebase}) is marked. All nodes are a priori indistinguishable
except from the degree and the label.
We assume
that the mobile agent is a Turing machine (with unbounded local
memory).
Moreover we assume that an agent accesses its memory and computes instructions 
instantaneously.
An execution $\rho$ of an algorithm $\A$ for a mobile agent is
composed by a (possibly infinite) sequence 
of moves by the agent.
The length $|\rho|$ of an execution $\rho$ is the total number of moves.




\subsection{The Exploration Problem}

We consider the classical exploration Problem for a mobile agent.
An algorithm $\A$ is an exploration algorithm if for
any graph $G=(V,E)$ with binocular labelling, for any port numbering $\delta_G$,
starting from any arbitrary vertex $v_0\in V$, 
the agent visits every vertex at least once and terminates.


We say that a graph $G$ is \emph{explorable} if there exists an Exploration
algorithm that halts on $G$ starting from any point.
An algorithm \A explores a family of graphs \gfam if it is an Exploration algorithm such that
for all $G\in\gfam$, \A halts and for all $G\notin \gfam$, either \A halts and explores $G$, either \A never halts; we say that it is a universal
exploration algorithm for \gfam.
We require the Exploration algorithm to not use any metric information
about the graph (like the size).

\section{Definitions and Notations} \label{sect:def}
\makeatletter{}
\subsection{Graphs} 
We always assume simple and connected graphs.  The following
definitions are standard \cite{Ros00}. Let $G$ be a graph, we denote
$V(G)$ (resp. $E(G)$) the set of vertices (resp. edges). If two
vertices $u,v\in V(G)$ are adjacent in $G$, the edge between $u$ and
$v$ is denoted by $uv$.

\parag{Loops, Paths and Cycles} 
A loop in a graph $G$ is a  sequence of
vertices $(v_0,...,v_k)\subseteq V(G)$ such that either $v_iv_{i+1}\in E(G)$, either $v_i=v_{i+1}$,
 for every $0\leq i < k$.
The length of a loop  is equal to the number of vertices composing it. 
A path $p$ in a graph $G$ is a loop $(v_0,...,v_k)$ such that $v_iv_{i+1}\in E(G)$ for every
$0\leq i < k$.
We say that the length of a path $p$, denoted by $|p|$, is the number
of edges composing it. 
We denote by $p^{-1}$ the inverted sequence of $p$. 
A path is simple if for any $i\neq j$, $v_i\neq v_j$.
A cycle is a path such that $v_0=v_k$, $k\in\N$. 
A cycle is simple if the
path $(v_0,\dots,v_{k-1})$ is simple or it is the empty path.
On a graph endowed with a port numbering, a path $p=(v_0,...,v_k)$ is 
labelled by $\lambda(p)=(\delta_{v_0}(v_{1}),\delta_{v_1}(v_2),...,\delta_{v_{k-1}}(v_k))$.

The distance between two vertices $v$ and $v'$ in a graph $G$ is  
denoted by $d_G(v,v')$.
It is the length of the shortest path between $v$
and $v'$ in $G$.
The set $pred_{v_0}(v)=\{u\mid uv\in E(G)\land d(v_0,u)=d(v_0,v)-1\}$ is called
the set of predecessors of $v$.
We denote it by $pred(v)$ if the context permits it.

We define $N_G(v,k)$ to be the subset of vertices of $G$  
 at distance at most $k$ from the vertex $v$ in $G$.
We define $B_G(v,k)$ to be the subgraph of $G$ 
induced by  $N_G(v,k)$.

Let $p$ be a path in a graph $G$ leading from a vertex $v$ to $w$.
We define $\destc_G:V(G)\times \N^\N$ such that $\dest{G}{v}{p}=w$, that is, $\dest{G}{v}{p}$ 
is the vertex in $G$ reached by the path labelled by $\lambda(p)$ starting from  $v_0$.

\parag{Layering partition and Clusters}

A \emph{layering} of a graph $G=(V,E)$ having a distinguished vertex $v_0$ is a partition of $V$ into sphere $S^i=\{v\mid d(v_0,v)=i\}$, $\forall i=1,2,\dots$
A \emph{layering partition} of $G$ is a partition of each $S^i$ into \emph{clusters} $C^i_1,\dots,C^i_p$ such that for every two vertices $u,v\in S^i$, $u$ and $v$ belong to $C^i_j$ if and only if there is a path $p$ from $u$ to $v$ passing inside $S^i$.
We denote by $d(v_0,C)$ the distance $h$ between vertices in $C$ and $v_0$.

We define below $\mapCCC(G)$, the graph of clusters of a graph $G$.
\begin{mydef}$ \mapCCC(G)=(V,E)$ such that 
  \begin{compactitem}
  \item $V= \{$cluster $C$ of $G\}$
  \item $E=\{CC'\mid \exists v\in V(C), \exists v'\in V(C')$ and $vv'\in E(G)\}$
  \end{compactitem}
\end{mydef}

The set of \emph{predecessors} of a cluster $C$ in $G$, denoted by $pred(C)$,
is composed by every cluster $C'$ in $G$ such that $C'C\in E(\mapCCC(G))$ and $d(v_0,C)=d(v_0,C')+1$.
Respectively, the set of \emph{successors} of $C$, denoted by $succ(C)$ is composed by every component $C'$ such that $CC'\in E(\mapCCC(G))$ and $d(v_0,C')=d(v_0,C)+1$.

\parag{Binoculars}
Our agent can use ``binoculars'' of range 1, 
that is, located on a vertex $u$, it can ``see'' the induced ball $B(u,1)$ (with port
numbers) of radius 1 centered on $u$.
To formalize what the agent sees with its
binoculars, 
we always assume that every graph $G$ is endowed with an 
additional vertex labelling $\nu$, called \emph{binoculars labelling} 
such that for any vertex $v\in V(G)$,  $\nu(v)$ is a graph isomorphic to $B(v,1)$ 
endowed with a port numbering $\tau$ induced from $G$.
Moreover, the agent is endowed with a primitive called $\getBino()$ permitting it 
to access to the binoculars labelling $\nu(u)$ of the vertex $u$ currently visited, 
that is, located on $u$, $\getBino()$ returns $B(u,1)$.
Note that the agent knows which is the explored vertex $u$ in $B(u,1)$.


\parag{Coverings}
We now present the formal definition of graph homomorphisms that capture the relation between graphs that locally look the same in our model.
A map $\varphi: V(G) \rightarrow V(H)$ from a graph $G$ to a graph $H$
is a \emph{homomorphism} from $G$ to $H$ if for every edge $uv \in
E(G)$, $\varphi(u)\varphi(v) \in E(H)$.
  A homomorphism $\varphi$ from $G$ to $H$ is a
  \emph{covering} if for every $v \in V(G)$, $\varphi_{\mid N_G(v)}$ is
  a bijection between $N_G(v)$ and $N_H(\varphi(v))$.

This standard definition is extended to labelled graphs
$(G,\delta,label)$ and $(G',$ $\delta',label')$ by adding the conditions
that $label'(\varphi(u))=label(u)$ for every $u\in V(G)$ and that
$\delta_u(v) = \delta'_{\varphi(u)}(\varphi(v))$ for every edge $uv \in
E(G)$. We have the following equivalent definition when $G$ and $G'$
are endowed with a port numbering.

\begin{prop}
\label{covering}
  Let $(G,\delta,label)$ and $(G',\delta',label')$ 
  be two labelled graphs, an homomorphism $\varphi:G\longrightarrow G'$
  is a \emph{covering} if and only if
  \begin{compactitem}
    \item for all $u\in V(G)$, $label(u)=label'(\varphi(u))$,
    \item for all $u\in V(G)$, $u$ and $\varphi(u)$ have same degree.
    \item for any $u\in V(G)$, for any $v\in
      N_G(u)$, $\delta_u(v) = \delta'_{\varphi(u)}(\varphi(v))$.
  \end{compactitem}
\end{prop}

\begin{prop}[Universal Cover]
  For any graph $G$, there exists a possibly infinite graph
  (unique up to isomorphism) denoted $\Gu$ and a covering
  $\mu:\Gu\to G$ such that, for any graph $G'$, for any
  covering $\varphi: G'\to K$, there exists a covering
  $\gamma:\Gu\to G'$ and $\varphi\circ\gamma = \mu$.
\end{prop}

It also possible to have a notion of simplicial covering.
A \emph{graph covering} $\varphi: G \to G'$ is a simplicial covering
$\varphi: V(G)\to V(G')$ such that for any vertex $v\in V(G)$, $\nu(v)=\nu(\varphi(v))$
This notion capture the indistinguishable graphs for an agent endowed with Binoculars.
We get a the following definition for the simplicial universal cove,

\begin{prop}[Simplicial Universal Cover]
  For any graph $G$, there exists a possibly infinite graph
  (unique up to isomorphism) denoted $\Gu$ and a simplicial covering
  $\mu:\Gu\to G$ such that, for any graph $G'$, for any simplicial
  covering $\varphi: G'\to K$, there exists a simplicial covering
  $\gamma:\Gu\to G'$ and $\varphi\circ\gamma = \mu$.
\end{prop}

From standard distributed computability results
\cite{YKsolvable,BVanonymous,CGM12}, it is known that the structure of
the covering maps explains what can be computed or not.
So in order to investigate the structure induced by coverings of
graphs with binoculars labelling, we will investigate the structure of
simplicial coverings. 
 We call simply ``coverings'' the simplicial
coverings. 

Note that the simplicial universal cover (as a graph with
binoculars labelling) can differ from its universal cover (as a graph
without labels). Consider for example, the triangle network.
\parag{Homotopy}

We say that two loops  $c =
(v_0,v_1,\ldots,v_{i-1},v_i,v_{i+1},\ldots,,v_k)$ and $c'=
(v_0,v_1,\ldots,v_{i-1},v_{i+1},\ldots,,v_k) $ in a graph $G$ are
related by an \emph{elementary homotopy} if one of the following
conditions holds (definitions from \cite{BH}):
\begin{compactenum}[(i)]
\item[(Contracting)] $v_i = v_{i+1}$,
\item[(Backtracking)] $v_{i-1} = v_{i+1}$, 
\item[(Pushing across a 2-cell)] $v_{i-1}v_{i+1}$ is an edge of $G$. 
\end{compactenum}

Note that being related by an elementary homotopy is a reflexive
relation (we can either increase or decrease the length of the loop).
We say that two loops $c$ and $c'$ are homotopic equivalent if there
is a sequence of loops $c_1,\ldots,c_k$ such that $c_1=c$, $c_k=c'$,
and for every $1\leq i<k$, $c_i$ is related to $c_{i+1}$ by an
elementary homotopy.  A loop is $k-$contractible (for $k\in\N$) if it
can be reduced to a vertex by a sequence of $k$ elementary homotopies.
A loop is contractible if there exists $k\in\N$ such that it is
$k-$contractible.

\parag{Simple Connectivity}
A \emph{simply connected} graph is a graph where
every loop can be reduced to a vertex by a finite sequence of
elementary homotopies. 

This definition is 
the graph version of the simplicial covering defined in \cite{CGN15} for simplicial complexes.
Simply connectivity have a lot
of interesting combinatorial and topological properties.
In our proofs, we rely on
the following fundamental result below.  
Even if this results applied for simplicial complexes, 
we can prove
that it holds for graphs and simplicial graph covering as defined above.
\begin{prop}[\cite{Lyn77}]
 Let $(G,\nu)$ be a connected graph with binoculars labelling,
 then  $G$ is isomorphic to $\Gu$, the universal simplicial covering
 of $G$ if and only if $G$ is simply connected.
\end{prop}

In fact, in order to check the simple connectivity of a graph $G$,
it is enough to check that all its simple cycles are contractible. 
The proof is straightforward and presented in Appendix.
We get the following Proposition,
\begin{prop}\label{prop:cycle_to_loop}
A graph $G$ is simply connected if and only if every simple cycle is
contractible.
\end{prop}

In Figure~\ref{img:SC}, we present two examples of simplicial covering
maps, $\varphi$ is from the universal cover, and $\varphi'$ shows the general property of coverings that is that the number of vertices of the bigger graph is a multiple (here the double) of the number of vertices of the smaller graph.

\begin{figure}[t]
  \caption{Simplicial Covering}
  \label{img:SC}
  \centering
  \includegraphics[height=2cm]{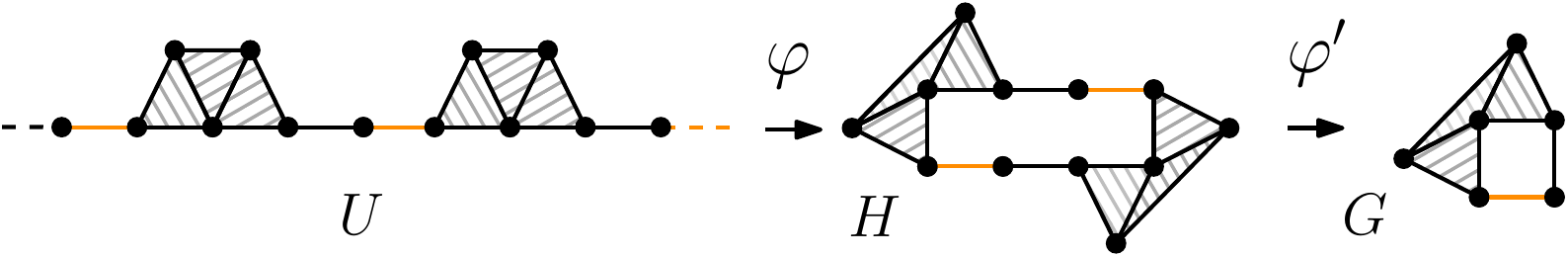}
\end{figure}

\parag{Weetman Graphs} We present now the family of graphs
investigated in this paper.

\begin{mydef}[Weetman Graphs \cite{weetman_locally_94}]
  A graph $G$ is Weetman if it satisfies for every vertex $v_0$ the following properties:
  \begin{compactenum}[(i)]
  \item[($\TCond$)] \textbf{Triangle Condition.} \label{TCond} 
    for every adjacent nodes $v,v'\in V(G)$ at distance $k$ from $v_0$, 
    there is a vertex $u$ at distance $k-1$ from $v_0$ such that $uv,uv'\in E(G)$.
  \item[($\IntCond$)] \textbf{Interval Condition.}  \label{IntCond}
    For every $v\in V(G)$, the subgraph induced by $\pred_{v_0}(v)$, the predecessors of $v$, is connected.
  \end{compactenum}
\end{mydef}
The family of Weetman graphs contains chordal graphs but also
non-chordal graphs like Johnson graphs \cite{TitsBuildings}.
Johnson graphs are graphs
whose vertices are subset of $k$ elements of a set with $n$ elements,
and whose edges link subsets that can be obtained from one another by
removing and adding one element (see Fig. \ref{fig:jonhson} for
instance).

They belong to the family of Simply connected graphs. 

\section{Properties of Weetman Graphs}
\label{sec:weetmanprop}

\begin{figure}[t]
  \caption{Johnson Graph $J(5,2)$} 
  \label{fig:jonhson}
  \centering
  \includegraphics[height=3cm]{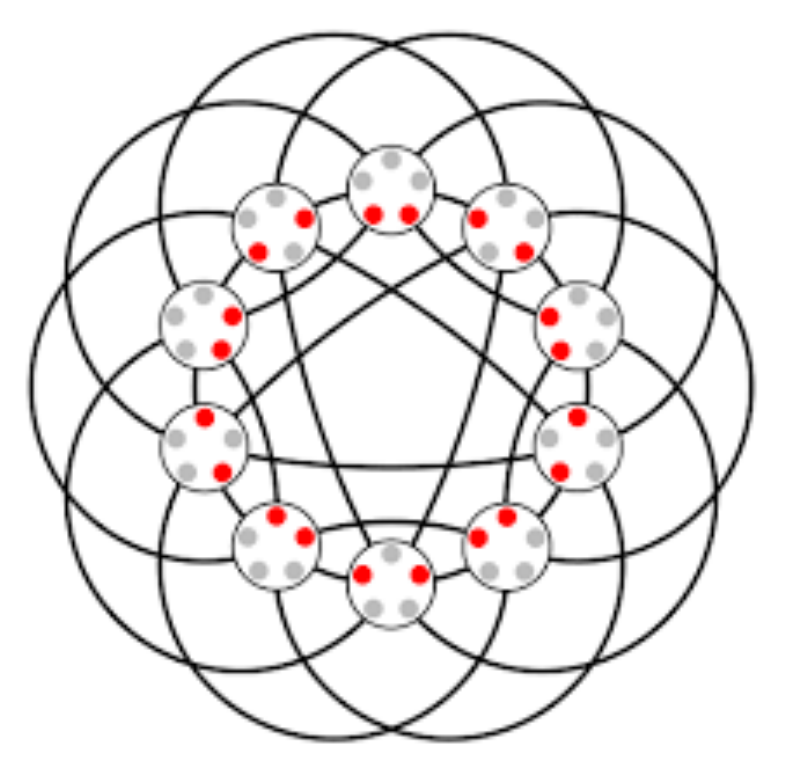}
\end{figure}

\begin{prop}\label{prop:weetman-SC}
  \begin{figure}
   \caption{Illustration of different cases of the proof of Lemma \ref{prop:weetman-SC}}
    \label{fig:proofSC}
    \centering
    \includegraphics[height=3cm]{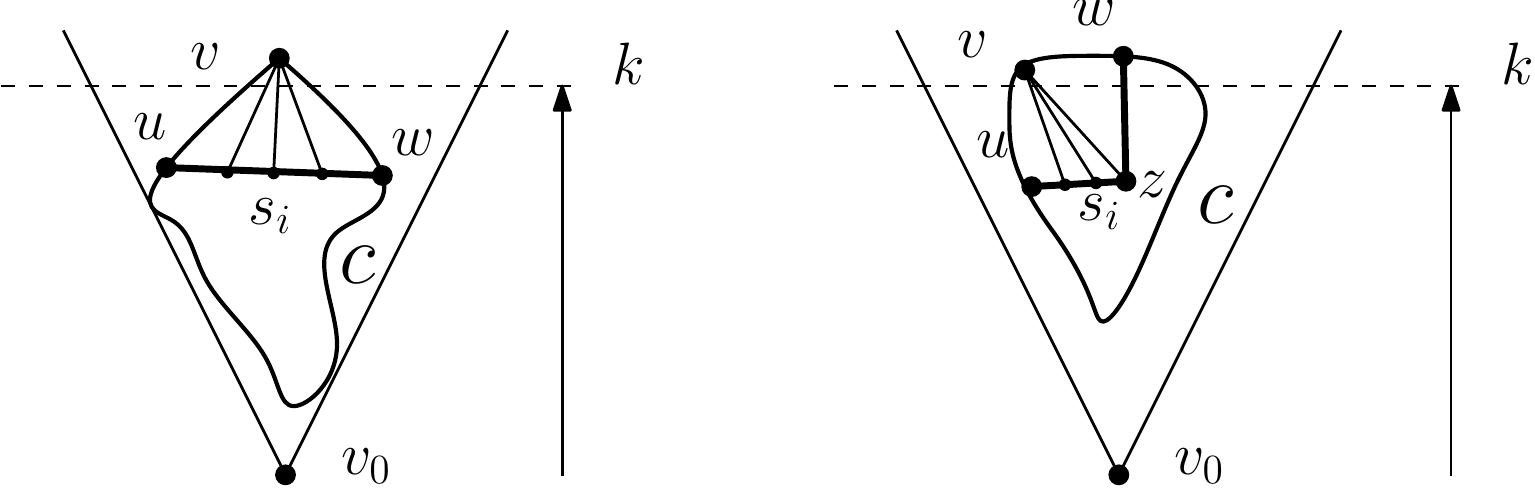}
  \end{figure}

 Weetman graphs are simply connected
\end{prop}
\begin{proof} 
Suppose by contradiction that there is one not contractible cycle $c$ in a Weetman graph $G$.
Among all not contractible cycle $c'$ in $G$, choose the cycle $c$ minimising $k=d(v_0,c)=max \{d(v_0,v')\mid v'\in c\}$ and minimising $|c\cap S^k|$.

We prove that there is another cycles minimising either $d(v_0,c)$, either $|c\cap S^k|$ which is not contractible in $G$, contradicting our hypothesis.

 There are two cases, if $|c\cap S^k|=1$, then there are $u,v,w\in V(c)$ such that $u,w\in S^{k-1}$, $v\in S^k$, $uv, vw\in E(G)$.
Moreover, Since $c$ minimises $d(v_0,c)$, there is no edge $uw\subset E(G)$.

Since $G$ is Weetman, from Condition $\IntCond$ applied on the triple of vertices $u,v,w$, 
there is a path $p=(s_1...s_n)\subset S^{k-1}$ such that
$s_1=u,s_n=w$ and for every $1\leq i\leq n$, $s_iv\in E(G)$.

Since the cycle $c'=c\setminus \{uvw\}\cup \{p\}$ is related to $c$ by a sequence of elementary homotopies, $c'$ is also not contractible and since $d(v_0,c')=k-1$.
We get a contradiction on the choice of $c$. 

If $|c\cap S^k|>1$, let $u\in S^{k-1}$, $v,w\in S^k$ such that $u$ and $w$ are the previous and consecutive vertex of $v$ in $c$, i.e. $uv,vw\in E(c)$.  


From Condition $\TCond$ (on the edge $vw$), there is a vertex $z\in S^{k-1}$ 
such that $zvw\in E(G)$ is a triangle.

From Condition $\IntCond$ (on the triple of vertices $uvz$),
there is a path $q=\{s_0\dots s_\ell\}\subset S^{k-1}$ such that $s_0=u$, $s_m=z$ and $s_iv\in E(G), \forall 0\leq i\leq \ell$.

Consequently, there is a cycle $c'=c\setminus \{uvw\} \cup \{qw\}$ passing thought vertices $uqw$ instead of $uvw$. 
Moreover, since $c'$ is related to $c$ by a sequence of elementary homotopies, $c'$ is also not contractible.

Consequently, since $|c'\cap S^k| < |c\cap S^k|$ and $c'$ is not contractible, we get a contradiction on the choice of $c$.

\end{proof}

The following Lemma prove that clusters in a simply connected graph have a particular topology that permit us later in this document to explore Weetman Graphs in a quasi optimal way.
Since the proof uses combinatorial tools which are not needed in the remaining part, 
the proof is left to the appendix.

\begin{thm}\label{SC-tree-clusters}
  Clusters of a simply connected graph form a tree 
\end{thm}

Consequently, from Theorem \ref{SC-tree-clusters}, we get trivially the next corollary.
\begin{cor}\label{Weetman-tree-clusters}
  The clusters of a Weetman graph form a tree.
\end{cor}
Note that this tree can be reduced to a path in some cases, like Johnson graphs or triangulated sphere.
See Fig.~\ref{fig:circularcomp}. 

\begin{figure}[t]
  \caption{Clusters in a Weetman graph}
  \label{fig:circularcomp}

  \centering
  \subfloat[General shape]{\includegraphics[width=6.5cm]{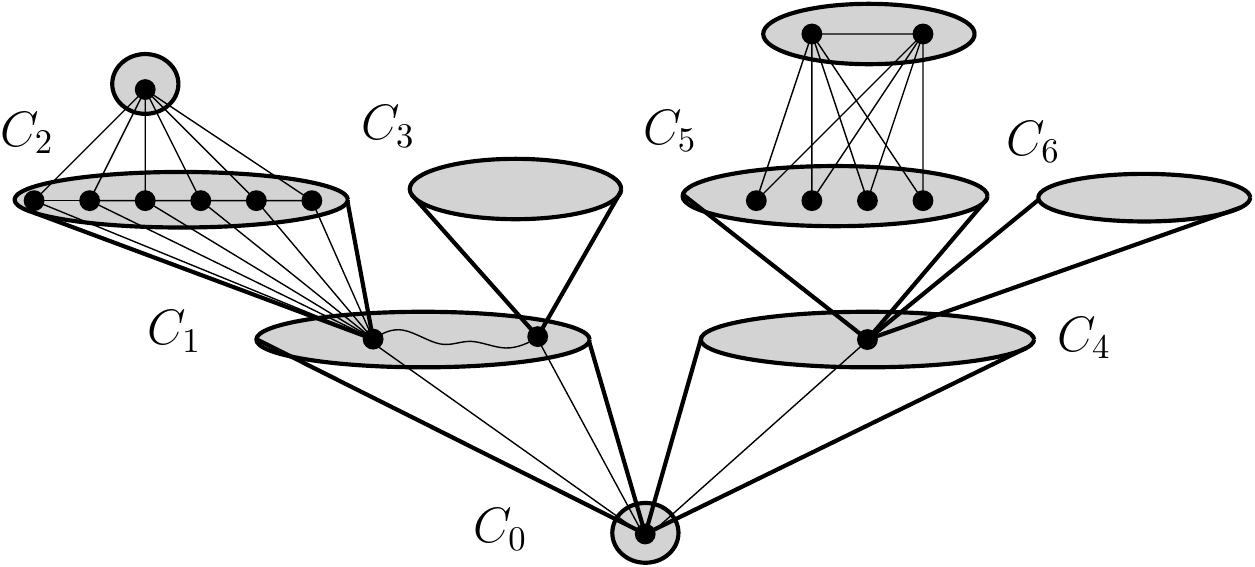}}~
  \subfloat[$Johnson~graph~ J(4,2)$]{\includegraphics[height=3.5cm]{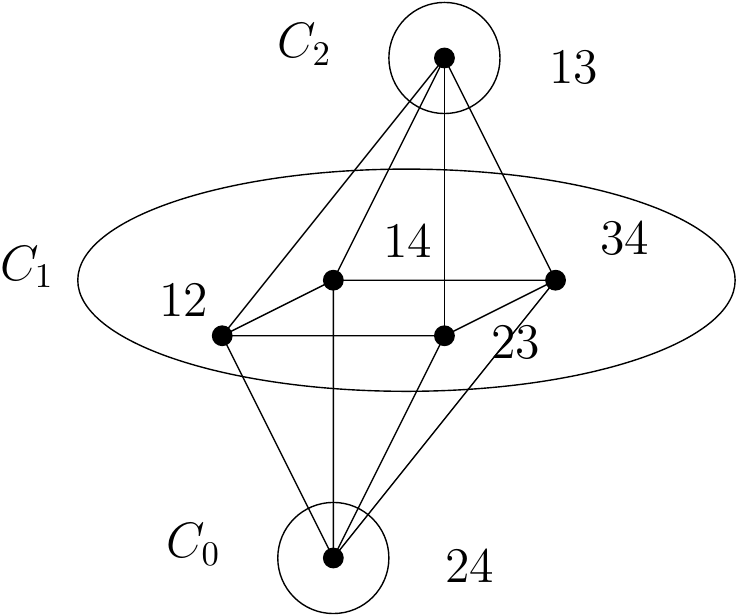}}
\end{figure}

From Corollary \ref{Weetman-tree-clusters}, every cluster $C\subset S^k$ of a Weetman graph $G$
admits a unique ancestor cluster, called $a(C)$, in the graph.
Note that $a(C)\subset S^{k-1}$ by definition and for the cluster $C_0$ containing $v_0$, $pred(C_0)=C_0$.

This particularity of Simply connected graph are the basis of the following exploration algorithm.
As illustrated in the next corollary, such a property eases the exploration and 
enables a ''quasi'' optimal exploration algorithm. 

\begin{cor}\label{prop:convex}
  Every path leading from the homebase $v_0\in V(G)$ of the agent 
  to a connected cluster $C$ goes through $a(C)$.
\end{cor}

\section{Weetman-universal Exploration Algorithm}


Exploring in anonymous graph is difficult. Because of the lack of ids attached to a
node, it might not be possible to
know whether a node where the agent is located is actually new. 
We introduce the following terminology. A node is \emph{explored} if
the agent has already been to this node. A node is \emph{discovered} if it
has been seen from another node (that is, if it is adjacent to an
explored node). 
In the following, we give a description of Algorithm~\ref{algo:weetman}.
In a nutshell, by a DFS traversal of the tree of clusters, for every
cluster $C$, the agent will explore 
all nodes of $C$, discovering in this way all child 
clusters of $C$. 
In this way, the agent is able to explore clusters in a DFS fashion. 
\medskip

Algorithm~\ref{algo:weetman} is divided into phases.
Between phases, the agent navigates the tree of clusters in
a DFS fashion using a stack.
In each phase, the agent explores a cluster
and updates local structures (more details below).
At the end of the phase, the agent extends its map with the new
identified vertices and corresponding edges.  
From the updated map, the agent computes the new clusters that appeared in $\mapG$.
We now give more details on the computations.

\parag{The map \mapG}
The map computed by the agent is denoted by $\mapG$.
A vertex $v\in V(G)$ is represented in $V(\mapG)$ by an unique integer
$n$ when the agent identify the vertex $v$.
The port numbering of map $\mapG$  is induced by the  port numbering
$\delta$ of the network $G$.


\parag{Local structures for the identification}

Let $p$ be the path followed by the agent from the
beginning of the exploration starting at $v_0$.
Let $u$ be the current location vertex of the agent in
$G$ corresponding to $n$ in its map $\mapG$ and 
let $n_0$ be the vertex corresponding to the homebase of the agent in $\mapG$.
That is, $\dest{G}{v_0}{p}=u$ and $\dest{\mapG}{n_0}{p}=n$. 

First, remark that the agent exploring the graph knows in
 its map on which vertex $n=\dest{\mapG}{n_0}{p}$ it is located.



During the exploration of a cluster $C$, for every explored vertex
$u\in V(G)$ identified by $n\in V(\mapG)$.
the agent looks through its binoculars, that is, it calls $\getBino()$ 
and access to $B_G(u,1)$, 
the ball of radius $1$ around $u$ (Line \ref{line:getBino}). 

The ball $B_G(u,1)$ obtained is stored into the set $\Balls$ in order to, 
at the end of the phase, update and verify 
the ''local'' correctness of the map $\mapG$.
We denote by $\Balls[n]$ the ball obtained after calling $\getBino()$ 
on a vertex $n\in V(\mapG)$.
Since the agent have to know which vertex in $\Balls[n]$ corresponds 
to its current location, we
introduce $\psi(n)\in \Balls[n]$ 
to denote the vertex corresponding to $n$ in $\Balls[n]$.
 
At the end of the phase, i.e., the end of the exploration
of the cluster, the agent, using $\Balls$, updates two data
structures that are used to identify the nodes and update $\mapG$.

The first structure $\prevert$ encodes the existence of  
new vertices that are not present in $\mapG$ 
as seen by the agent from an explored node $n$ at the
phase $i$.

Since such new vertices are linked to a vertex explored phase $i$, 
we encode a vertical edge by a tuple $(n,p,q)$ where $n$ is the id of
the explored vertex and $(p,q)$ is the labelling of the vertical edge.
We call \emph{pre-vertex} a pair $(n,p)$.
Pre-vertices give us all newly discovered vertices.
Note that $q$ is stored along the pre-vertex $(n,p)$ only to simplify
the computation of \mapG at the end of the phase.

The main idea to correctly map newly discovered vertices is, at the end of a phase,  
to find all vertical edges pointing to this node. 
So, we add a second structure,
denoted by $\Requiv$, encoding the elementary relation between
two pre vertices corresponding to a same vertex in $G$.
So if there is $((n,p),(m,q))\in\Requiv$, then 
there is a couple of pre-vertices
$(n,p),(m,q)\in\prevert$ such that $n$ and $m$ are the 
ids of vertices explored during the phase, 
$nm$ is an edge of $E(\mapG)$ and there is
$uvw\subseteq \Balls[n]$ such that $u=\psi(n)$, 
$\lambda(nm)=\lambda(uv)$ and $\destjc{\mapG}{n}{\delta_u(w)}$
is not defined in $\mapG$.

In order to update $\mapG$ in such a way that 
all edges between discovered nodes are correctly mapped, we have to
distinguish two kind of edges. There are edges between an explored
node and a newly discovered node. 
They are called ``vertical'' edges.
But edges between two discovered nodes are also to be correctly
mapped. These edges are called ``horizontal'' edges.

To gather ''vertical'' edges, $\equiv$ is sufficient as explained below.
To gather ''horizontal'' edges, the set $\nivsup$ is introduced.
Since it is not possible to identify ''horizontal'' edges before 
the end of the exploration of the cluster, we store in $\nivsup$, 
together with the port numbers associated to the horizontal edge, 
only the pre-vertices.
Namely, elements of $\nivsup$ are tuples $(n,p_1,p_2,(r,s))$ such that
there is a triangle $uvw$ in $\Balls[n]$ where $u=\psi(n)$
and there is two pre vertices $(n,p_1),(n,p_2)\in \prevert$ such that
$\delta_{u}(v)=p_1$, $\delta_{u}(w)=p_2$ and $\lambda(vw)=(r,s)$.
Note that by definition, $\destjc{\mapG}{n}{p_1}$ and $\destjc{\mapG}{n}{p_2}$ 
are not defined in $\mapG$ during the exploration of the cluster. 

\parag{Updating the map}

First remark that the transitive and reflexive closure of $\equiv$, 
denoted by  $\cloequiv$, is an equivalence relation 
between pre vertices (it is straightforward that $\equiv$ is reflexive). 
We denote by $[n,p]$, the representative (or the equivalence class) 
of a class of pre-vertices including $(n,p)$ via $\cloequiv$.

To identify new nodes ( Line \ref{line:newvertex}), 
we compute the quotient of the relation $\cloequiv$ over pre-vertices
stored in $\prevert$. 
Then, for every equivalence classes $[n,p]$ of pre-vertices that arises, 
we add a new vertex $\newvert{n,p}$ in $\mapG$.
Moreover, every node $n\in V(\mapG)$ is endowed with an additional 
label $\cir(n)\in\N$ to store the identity of the cluster including $n$.

Then, for every pre vertex $(n,p,q)\in\prevert$,
we add a ''vertical'' edge linking $n$ to $\newvert{n,p}$ labelled by $(p,q)$ 
if the edge is not already present in $\mapG$.

To update the ''horizontal'' edges of $\mapG$, 
for every  $(n,p_1,p_2,(r,s))\in \nivsup$, 
we add an edge between $\newvert{n,p_1}$ and $\newvert{n,p_2}$ labelled by $(r,s)$
if the edge is not already present in $\mapG$.

Additionally, if the agent  ends phase $i$, for every vertex
$n\in V(\mapG)$ explored during this phase, 
$\vis(n)$ is set to $i$.

Once $\mapG$ has been updated, we compare the map obtained and
what we saw during the exploration of the cluster.
That is, for every vertex $u$ corresponding to $n$ explored phase $i$,
we compare $B_{\mapG}(n,1)$ and $\Balls[n]$, 
the binoculars labelling obtained from $u$.
If we detect an error in the map, we decide to continue the exploration 
forever in order to respect the exploration specification.
This case will be more discussed later in this document.
If no error is detected, 
the new clusters that appear in $\mapG$ are computed, numbered, and push to the stack $\stac$ at Line \ref{line:newcomp}.




The agent stops its exploration when it remains no cluster to
explore, that is, when
$\stac$ is empty. 


\begin{rem}
  The only cluster at level $0$ is composed of the homebase of the agent $v_0$.  
\end{rem}

\SetKwComment{tcc}{/* }{ */}
\SetCommentSty{it}  

\begin{algorithm}
  \DontPrintSemicolon
  \LinesNumbered
  \BlankLine
  \SetKwFunction{algori}{algori}\SetKwFunction{proc}{proc}
  \SetKwProg{Proc}{Procedure}{}{}
  \SetKwProg{Algo}{Main Procedure}{}{}
  \caption{Weetman Graphs Exploration \label{algo:weetman}}
  \Algo{}{
    
    \BlankLine
    \textit{Add $\n_0$ to $V(\mapG)$\label{line:homebase}}\;
    \textit{$\stac \gets\cir(\n_0)\gets \vis(\n_0)\gets 0$}\tcc*[r]{homebase cluster}\;
    \BlankLine 
    \While{$\stac\neq\emptyset$}{
      \BlankLine      
      \tcc*[r]{Beginning of the phase $i$}
      \textit{$\idcomp\gets$ top element of $\stac$}\;
      \textit{$\cluster\gets \{\n\in V(\mapG)\mid \cir(\n)=\idcomp \}$}\;
      \BlankLine
      \textit{Compute a shortest path $\parc$ which from the current
        location $\m$ of the agent, explores the cluster $\cluster$}\;     
      \BlankLine      
      \ForAll{$\n\in \cluster$      \tcc*[r]{Cluster exploration}}{
        \textit{Go to the vertex \n following $\parc$}\;
        \textit{$\vis(\n)=i$\label{line:id}}\;
        \textit{$\Balls[n]\overset{\cup}{\gets}\getBino(\n)$}\tcc*[r]{Looking thought Binoculars}\label{line:getBino}\;  
      }               
      \ForAll{$\n\in \cluster$\label{line:updateMap}}{
        \textit{Get $\Balls[\n]$ from $\Balls$ and let $u\in V(\Balls[\n])$ corresponding to $\n$}\;
        \tcc*[r]{new pre vertices}
        \ForAll{$uw \in E(\Balls[\n])$ s.t. there is no adjacent
          edge to $\n$ in $\mapG$ labelled by $\lambda(uw)=(\delta_{u}(w), \delta_{w}(u))$}
        {
          $\prevert\overset{\cup}{\gets} (\n,\delta_{u}(w), \delta_{w}(u))$\label{line:newprevertex}\;
        }
        \ForAll{triangle $uvw \subseteq \Balls[\n]$}{
          \tcc*[r]{new pre vertices/ vertical relation}
          \If{there is an edge labelled $\lambda(uv)$ and
            there is no edge labelled $\lambda(uw)$ and
            $\lambda(vw)$ adjacent to $\n$ in $\mapG$\label{line:newrelation}}
          {
            Let $m$ be the vertex in $\mapG^i$ such that
            $nm\in E(\mapG)$ and $\lambda(nm)=\lambda(uv)$\;
            $\equiv \overset{\cup}{\gets}$
            $\Big(\big(\n,\delta_{u}(w), \delta_{w}(u)\big),
            \big(\m,\delta_{v}(w), \delta_{w}(v)\big)\Big)$\;
          }
          
          \tcc*[r]{new horizontal edge relation}
          \If{there is no edge labelled by $\lambda(uv),\lambda(uw)$
            adjacent to $\n$ in $\mapG$\label{line:new-hor-relation}}{
            \textit{$\nivsup\overset{\cup}{\gets}
              \big(\n,\delta_{u}(v),\delta_{u}(w),(\delta_v(w), \delta_{w}(v))\big)$}\;
          }
        }
      }
        \BlankLine
          \tcc*[r]{Updating $\mapG$}
        \ForAll{$[\n,\p,\q]\in \quotient$}{
          \textit{Add a new vertex $\newvert{\n,\p}$ to $\mapG$\label{line:newvertex}\;}
          \textit{$\vis(\newvert{\n,\p})=\perp$\label{line:perp}}\;
        }
        \ForAll{$(\n,\p,\q)\in \prevert$}{
          \textit{Add a new edge $\n\newvert{\n,\p}$ labelled $(\p,\q)$ to $\mapG$\label{line:verticaledge}}\;
        }
        \ForAll{$\big(\n,\p,\q,(\p',\q')\big)\in\nivsup$}{
          \textit{Add a new edge $\newvert{\n,\p}\newvert{\n,\q}$ labelled $(\p',\q')$ in $\mapG$\label{line:horizontaledge}}
        }
        \eIf{$\forall \n\in \cluster$, $\Balls[\n]$ is isomorphic to $B_{\mapG}(\n,1)$\label{line:checkIso}}{
          \textit{Push in $\stac$ the new clusters in \mapG\label{line:newcomp}} \;
          \textit{For every new vertex $\newvert{\n,\p}$, update $\cir(\newvert{\n,\p})$} 
        }{
          \textit{Continue forever the execution along an edge}
        }
        \tcc*[r]{End of the phase $i$}
      }
    }
  \end{algorithm}



  Let $\borderM=\{n\in V(\mapG)\mid \vis(n)=\perp\}$ be the subgraph induced 
  from $G$ composed by the set of vertices in $G$ discovered and
  not yet explored by the agent.

  Note that at $\mapG^i$ corresponds to the map computed at the end of the phase $i$. 

  First we prove the correctness properties, that is, if no error is detected 
  and if the algorithm halts then the graph is explored.
  Then, the proof of the termination of Algorithm \ref{algo:weetman} on Weetman graphs 
  will be straightforward.

  The core of correctness proof is the following theorem
  \begin{thm}\label{thm:inj-surj-map}
    For every phase $i$ of the algorithm, there is an homomorphism $\varphi:M^{i}\to G$ such that
    \begin{itemize}
    \item  for every $n\in V(\mapG^i)$, 
      $\varphi_{|N_{M^i}(n)}$ is injective
    \item  for every $n\in V(\mapG^i\setminus \borderM^i)$, 
      $\varphi_{|N_{M^i}(n)}$ is surjective
    \end{itemize}
  \end{thm}
  Theorem \ref{thm:inj-surj-map} is proved by an induction on the 
  phases perform by the agent during  the execution.
  
  The homomorphism $\varphi$ is based on
  the following corollary of Theorem \ref{thm:inj-surj-map},



  \begin{cor}
    For every phase $i$, For every vertex $n\in V(\mapG^i)$, 
    for every path $p$ from $n_0$ to $n$ in $\mapG^i$, 
    $\varphi^i(n)=\dest{G}{\varphi(n_0)}{p}$
  \end{cor}
  \begin{proof}
    Straightforward from Theorem 2
  \end{proof}

  To ease the notation, an homomorphism $\varphi:\mapG^i\to G$ is denoted by $\varphi^i$.

  \subsection*{Proof of theorem \ref{thm:inj-surj-map}}


We prove in the next Lemma the initial case of the induction.

  \begin{lem}[Initial Case]\label{goodmap:initial}
    For any execution of Algorithm \ref{algo:weetman} on a graph $G$
    endowed with a binoculars labelling,  
    $\mapG^1$ is isomorphic to $B_G(v_0,1)$.  
  \end{lem}
  \begin{proof}
    Since for every vertex $v\neq v_0$, $d(v_0,v)>0$, during the first phase,
    the agent have to explore the first cluster composed by only one vertex,  the homebase $v_0$.
    Since the agent is initially located on $v_0$, 
    the agent only maps the neighbourhood of the homebase
    $v_0$ during this phase.
    
    Initially, a first vertex $\n_0$ is inserted at
    Line \ref{line:homebase} into $\mapG^0$. 
    This vertex identifies the home base $v_0$.
    Let $\varphi(\n_0)=v_0$. 
    At Line \ref{line:getBino}, $B_G(v_0,1)$ is gathered into $\Balls[\n_0]$.
    Then, the agent updates its map (Line \ref{line:updateMap}). 
    \begin{itemize}
    \item First, since $G$ has an injective port numbering,
      there is a unique pre-vertex $(\n_0,\p,\q)$ 
      inserted into $\prevert$, for every edge
      $v_0w$ in $E(B_G(v_0,1))$ labelled $(\p,\q)$.
    \item Remark that there is exactly one pre-vertex per
      equivalence classes of pre-vertices.
      Thus, we get that $\quotient \simeq \prevert$.
    \item Consequently, since one vertex is added
      into $\mapG^1$ for each equivalence class, 
      there is a bijection between $V(\mapG^1)$ and $V(B_G(v_0,1))$. 
    \item Moreover, since there is also one 
      vertical edge inserted at Line \ref{line:verticaledge} 
      for each equivalence class, 
      there is a bijection between vertical edges in $V(\mapG^1)$ 
      and vertical edges in $V(B_G(v_0,1))$. 
    \item It remains to prove that there is also a 
      bijection between horizontal edge of $E(\mapG^1)$ and
      horizontal edge of $E(B_G(v_0,1))$.
    \item For every ''horizontal'' edge $ww'\in E(B_G(v_0,1))$ 
      i.e., $w'\neq v_0\neq w$, there are
      $n_0\newvert{n_0,\delta_{v_0}(w)},n_0\newvert{n_0,\delta_{v_0}(w')}\in 
      E(\mapG)$ from the previous case.
      Moreover, the couple
      $\lnivsup{n_0,\delta_{v_0}(w),\delta_{v_0}(w')}{\delta_{w}(w'),\delta_{w'}(w)}$ 
      is inserted into $\nivsup$ at Line \ref{line:new-hor-relation}.
    \item Thus, an edge
      $\newvert{n_0,\delta_{v_0}(w)}\newvert{n_0,\delta_{v_0}(w')}$ 
      labelled by $(\delta_{w}(w'),\delta_{w'}(w))$ is inserted
      into $E(\mapG)$ at Line \ref{line:horizontaledge} 
      if and only if there is an edge $ww'$ in $G$ labelled 
      by $(\delta_{w}(w'),\delta_{w'}(w))$.
    \end{itemize}
  \end{proof}

  Consequently, let us define the homomorphism $\varphi^1:\mapG^1\to G$ such that 
  $\varphi^1(n)=\destjc{G}{v_0}{\delta_{n_0}(n)}$, for every $n\in V(\mapG^1)$.
 Since $\varphi^1$ is an isomorphism (Lemma \ref{goodmap:initial}), we get that at phase $1$,
 Theorem \ref{thm:inj-surj-map} is proved.
Moreover, we get the following Corollary, 
\begin{cor}
  At phase $1$, for every vertex $n\in V(\mapG^{1})$, for every path $p$ from $n_0$ to $n$ in $\mapG^{1}$, 
  $\varphi^{1}(n)=\dest{G}{\varphi(n_0)}{p}$
\end{cor}

\subsection*{Phase i-1:  (Induction hypothesis)}
Suppose that the agent ends phase $i-1>0$ and the agent
has computed $\mapG^{i-1}$ such that no error is detected.
Moreover, suppose that there is an homomorphism $\varphi^{i-1}:\mapG^{i-1}\to G$ 
define as follows:
    \begin{itemize}
    \item If $m\in V(\mapG^{i-2})$, $\varphi^{i-1}(m)=\varphi^{i-2}(m)$
    \item If $m\in V(\mapG^{i-1}\setminus \mapG^{i-2})$, 
      \begin{itemize}
      \item Let $(n,p)\in\prevert$ such that $\newvert{n,p}=m\in V(\mapG^{i-1})$.
      \item Let $\varphi^{i-1}(m)=\destjc{G}{\varphi^{i-1}(n)}{\delta_n(m)}$
      \end{itemize}
    \end{itemize}    
    Note that by induction and since $\mapG^{i-2}\subseteq \mapG^{i-1}\subseteq \mapG^{i}$, 
    for every $n\in V(\mapG^{i-2})$, $\varphi^{i-1}(n)$ is well defined in $\mapG^{i}$.
    The following corollary explains that the image in $G$ of 
    a vertex in $\mapG^{i-1}$ via $\varphi^{i-1}$ is independent of
    the path followed by the agent.
    \begin{cor}
      At phase $i-1$, for every vertex $n\in V(\mapG^{i-1})$, 
      there is a vertex $w\in V(G)$ such that for every path $p$ 
      from $n_0$ to $n$ in $\mapG^{i-1}$, $\varphi^{i-1}(n)=\dest{G}{\varphi^{i-1}(n_0)}{p}$
    \end{cor}
We suppose that Theorem \ref{thm:inj-surj-map} is proved
at phase $i-1$, that is,  $\varphi^{i-1}$ is locally injective
from $\mapG^{i-1}$ and locally surjective from $\mapG^{i-1}\setminus \borderM^{i-1}$.

Remark that for every vertex $n\in \mapG^{i-1}\setminus \borderM^{i-1}$
($n$ already explored at phase $i-1$),
$B_{\mapG^{i-1}}(n,1)\simeq B_G(\varphi^{i-1}(n),1)\simeq \getBino(n)$.

  \subsection*{Phase i:}

  We prove that when the agent ends phase $i>1$ and has computed $\mapG^{i}$,
  there is an homomorphism $\varphi^i:M^{i}\to G$ such that
  \begin{itemize}
  \item  for every $n\in V(\mapG^{i})$, 
    $\varphi^i_{|N_{M^{i}}(v)}$ is injective
  \item  for every $n\in V(\mapG^{i}\setminus \borderM^{i})$, 
    $\varphi^i_{|N_{M^{i}}(v)}$ is surjective
  \end{itemize}

  Next Lemma prove that the relation $\equiv$ gathers in 
  a same equivalence class the maximum set of ''vertical'' 
  edges linking a same vertex in $G$. 
  Moreover, it ensure that the image of a vertex
  via $\varphi^i$ as defined above is independent
  of the choice of the representative $[n,p]$.
  \begin{rem}
    Note that in some graphs which are not Weetman, some vertices
    can be duplicated in $\mapG$. 
  \end{rem}

  \begin{lem}\label{lem:equiv-preserved}
  For every phase $i$ of Algorithm \ref{algo:weetman}, 
  for every pre-vertices $(n,p),(m,q)\in \prevert^i$, 
    if $(n,p)\cloequiv (m,q)$ then there is a vertex
    $w\in B_G(\varphi^{i-1}(n),1)\cap B_G(\varphi^{i-1}(m),1)$ 
    such that $\varphi^{i-1}(n)w, \varphi^{i-1}(m)w \in E(G)$ 
    and $\delta_{\varphi^{i-1}(n)}(w)=p$ and $\delta_{\varphi^{i-1}(m)}(w)=q$.  
  \end{lem}

  \begin{proof} 
    Suppose that there are two pre-vertices $(n,p),(m,q)\in \prevert^i$
    such that $(n,p)\equiv (m,q)$.
    \begin{itemize}
    \item From the algorithm, $n$ and $m$ are visited during phase $i$ .
    \item Moreover, there is a vertex $\newvert{n,p}=\ell\in V(\mapG^i)$ 
      linked to $n$ and $m$ such that $\delta_n(\ell)=p$ and $\delta_m(\ell)=q$.
    \item By induction, there is $u,v\in V(G)$ such that $u=\varphi^{i-1}(n)$, $v=\varphi^{i-1}(m)$,
      and $\varphi^{i-1}(nm)=\varphi^{i-1}(n)\varphi^{i-1}(m)$ is an edge in $E(G)$.
    \item Moreover, let $p_n:n_0\to n$ (resp $p_n:n_0\to m$) denotes
      the path follows by the agent from the beginning of the execution
      when it visits $n$ (resp. $m$) for the first time .
    \item By induction,  $\varphi^{i-1}(n)=\dest{G}{v_0}{p_n})=u$ 
      and $\varphi^{i-1}(m)=\dest{G}{v_0}{p_m})=v$, that is, 
      $u$ and $v$ are vertices where the agent 
      is located corresponding to $n$ and $m$ in its map.
    \item Since $(n,p),(m,q)\in \prevert^i$ and from from homomorphism definition, 
      the agent has seen two times the triangle $uvw$. Once inside $B_G(\varphi^{i-1}(n),1))$ 
      when it visits $n$ such that $\psi(n)=u$ and once inside
      $B_G(\varphi^{i-1}(m),1))$ 
      when it visits $m$ such that $\psi(m)=v$
    \item Consequently, $w\in B_G(\varphi(n),1)\cap B_G(\varphi(m),1)$ and 
      there is a vertex $w\in V(G)$ 
      such that $\varphi^{i}(u)w\varphi^{i}(m)$ is a triangle in $E(G)$, 
      we get the first case of this Lemma.
     \end{itemize}
    
    We now prove the case $(n,p)\cloequiv (m,q)$.
    \begin{itemize} 
    \item By definition, $(n,p)\cloequiv (m,q)$ implies that
      there is a sequence of pre vertices
      $(n_1,p_1),...,(n_k,p_k)\in \prevert$ such that 
      \begin{itemize}
      \item $(n,p)=(n_1,p_1)$ and $(m,q)=(n_k,p_k)$
      \item $\forall 1\leq h<k$, $(n_h,p_h)\equiv (n_{h+1},p_{h+1})$
      \end{itemize}
    \item Moreover, by induction, $\varphi^{i-1}(n_1,...,n_k)= \varphi^{i-1}(n_1)...\varphi^{i-1}(n_k)$ is a path in $G$.
    
    \item From the previous case, for every 
      $(n_h,p_h)\equiv (n_{h+1},p_{h+1})$, there is $w_h\in V(G)$ such that
      
\begin{itemize}
      \item  $\varphi^{i-1}(n_h)\varphi^{i-1}(n_{h+1})w_h \subset G$ 
        is a triangle 

      \item  $\delta_{\varphi^{i-1}(n_h)}(w_h)=p_h$ and $\delta_{\varphi^{i-1}(n_{h+1})}(w_h)=p_{h+1}$. 
      \end{itemize}
      
    \item From the transitivity of $\equiv$ and the injective
      port numbering function of $G$, we get that 
      $\delta_{\varphi^{i-1}(n_{h+1})}(w_h)=p_{h+1}=\delta_{\varphi^{i-1}(n_{h+1})}(w_{h+1})$.
    \item Consequently, we get that $w_h=w_{h+1}$ for every $1\leq h<k$.
    \item We prove that there is a unique vertex $w\in V(G)$ such that
      $\delta_{\varphi^{i-1}(n)}(w)=p $ and $\delta_{\varphi^{i-1}(m)}(w)= q $.
    \end{itemize}
  \end{proof}




  The above Lemma permits us to define the homomorphism $\varphi^i: V(\mapG^i) \to V(G)$ 
  such that for every $m\in V(\mapG^{i-1})$, 
  $\varphi^i(m)=\varphi^{i-1}(m)$ and  for every $m\in V(\mapG^{i}\setminus \mapG^{i-1})$, 
  $\varphi^i(m)=\destjc{G}{n}{p}$ such that $m=\newvert{n,p}$.

 It is straight forward to prove that $\varphi^i$ is well defined for vertices.
 So, we prove in the next lemma that $\varphi^i:\mapG^i\to G$ is an homomorphism, 
 that is, the image of an edge is an edge. 

  \begin{lem}\label{lem:morphism-vertex}
    At phase $i$, for every edge $nm\in E(\mapG^i)$,
    $\varphi^i(nm)\in E(G)$
  \end{lem}

  \begin{proof}

    

  \begin{figure}[t]
    \centering
    \includegraphics[height=3cm]{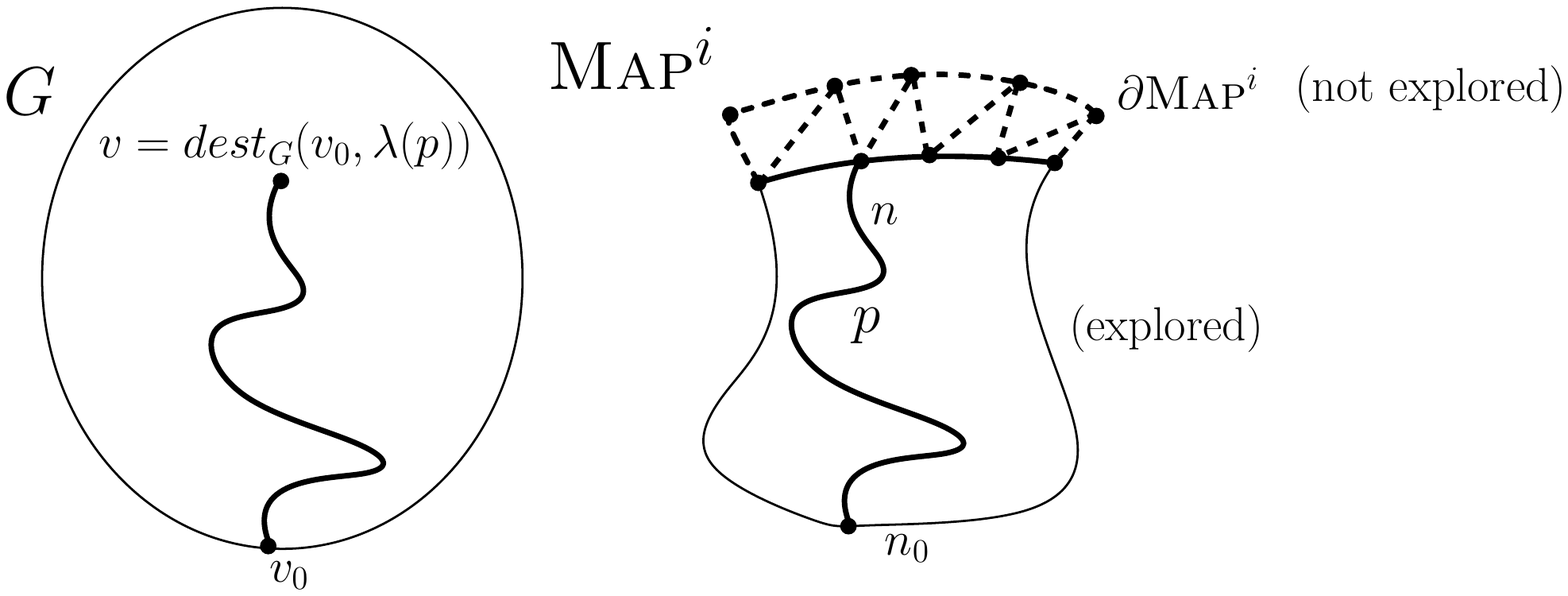}
    \label{img:map}
  \end{figure} 
    First, we prove the lemma for vertical edges
    and then, for horizontal edges.
    
    \parag{Vertical edges}
    \begin{itemize}
    \item We know that for every edge $nm\in E(\mapG^{i})$ such 
      that $n\notin V(\borderM)^i$ and $m\in V(\borderM)^i$ (vertical 'edge'), 
      there is a pre-vertex $[n,p]\in\prevert^i$ such that
      $\newvert{n,p}=m$.
    \item Moreover,  there is two vertices $u,v\in V(G)$ such that
      $\varphi^{i}(m)=v\neq \varphi^{i}(n)=u$ and $\delta_u(v)=p$.
    \item we get that $\varphi^{i}(nm)=\varphi^i(n)\varphi^i(m)=uv\in E(G)$
      is well defined for every ''vertical'' edge $nm\in E(\mapG^i)$ 
    \end{itemize}

    \parag{Horizontal edges}
    Now, we prove that  $\varphi^i$ correctly maps ''horizontal'' edges in $\mapG$ to $G$.
    \begin{itemize}
    \item By construction, every edge $\newvert{n,p}\newvert{m,q}\in E(\mapG^i\setminus\mapG^{i-1})$
      implies that there is $(n,p,p',(r,s))\in \nivsup$ such that 
      $(n,p),(n,p')\in \prevert$ are two pre-vertices and
      $(n,p)\not\equiv (n,p')\equiv (m,q)$ (injective port numbering function of $G$) .
    \item Moreover, the agent located on $u=\varphi^{i-1}(n)$ has seen a
      triangle $uww'\in B_G(\varphi^{i-1}(n),1)$
      such that 
      \begin{itemize}
      \item the edge $ww'$ is labelled $(r,s)$
      \item there is no $m,m'\in V(B_{\mapG^{i-1}}(n,1))$ such that $\delta_n(m)=p$ and $\delta_n(m')=p'$
      \end{itemize}
    \item From the previous case, $w=\varphi^i(\newvert{n,p})$ and
      $w'=\varphi^i(\newvert{m,q})$.
    \item Moreover, since $p\neq p'$, $w\neq w'$ and thus, $\varphi^i(n\newvert{n,p})=vw\neq \varphi^i(n\newvert{n,p'})=vw'$.
    \item Thus, $\varphi^i(\newvert{n,q}\newvert{m,q})=\varphi^i(\newvert{n,q})\varphi^i(\newvert{m,q})=ww'\in V(G)$ is well defined for every horizontal edge $\newvert{n,q}\newvert{m,q}$ in $E(\mapG^i)$
    \end{itemize}
  \end{proof}

  By induction, Theorem is already proved for every vertex
  included in $\mapG^{i-1}$ which are not explored phase $i$.
  Consequently, we only prove Theorem \ref{thm:inj-surj-map} for 
  \begin{itemize}
  \item vertices explored phase $i$ (Lemma \ref{lem:injective-in-map} and \ref{lem:surjective-in-map} ), i.e., that belongs to $\borderM^{i-1}\setminus \borderM^i$
  \item and newly discovered vertices (Lemma \ref{lem:injective-on-border}), i.e., that belongs to $\mapG^i\setminus \mapG^{i-1}$

  \end{itemize}

  \begin{lem}\label{lem:injective-in-map}
    for every $n\in V(\mapG^i\setminus \borderM^{i})$, for every edges $m_1m_1', m_2m_2'\in E(B_{\mapG^i}(n,1))$, if $m_1m_1'\neq m_2m_2'$ then $\varphi(m_1m_1')\neq \varphi(m_2m_2')$.
  \end{lem}
  \begin{proof}

  \begin{figure}[h]
    \centering
    \includegraphics[height=3.5cm]{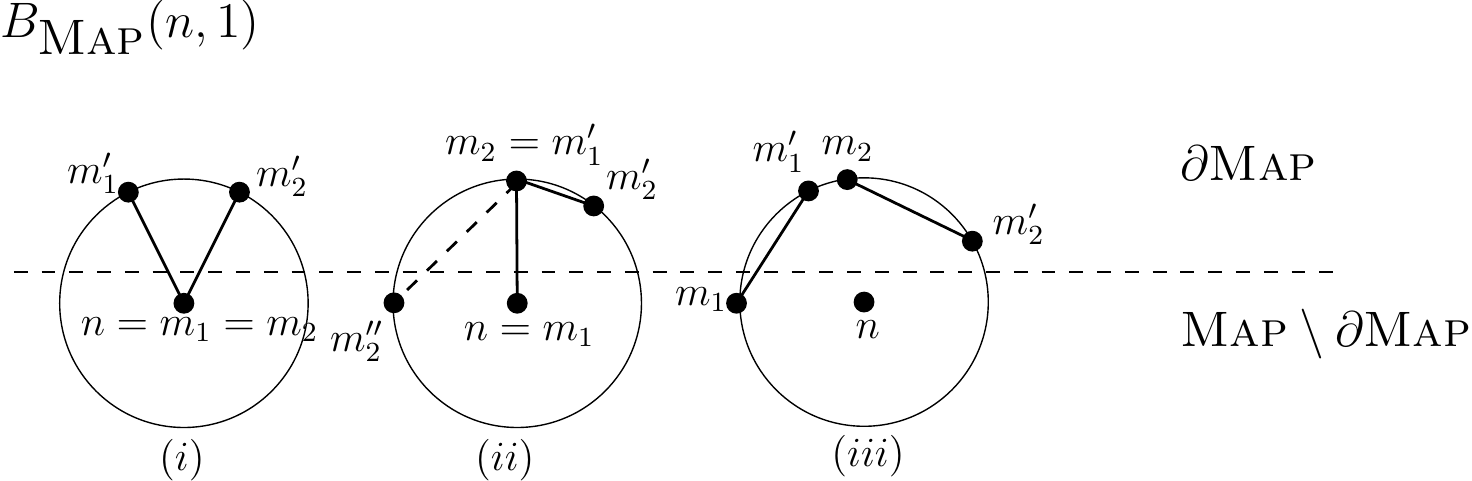}
    \caption{Different cases of Lemma 5.9}
    \label{img:lemma59}
  \end{figure} 
    Let $n\in V(\mapG^i\setminus \borderM^{i})$ be a vertex explored phase $i$ and let $m_1m'_1,m_2m'_2\in E(B_{\mapG^i}(n,1))$ such that $m_1m'_1\neq m_2m'_2$.
    Three cases appear, 
    \begin{enumerate}[i)]
    \item  $m_1=m_2=n$ and  $m_1'\neq m'_2$.
      \begin{itemize}
      \item Since $n\in V(\mapG^i\setminus \borderM^{i})$, by induction, $\Balls[n]\simeq B_G(\varphi^i(n),1)$
      \item Since $G$ and $\mapG^i$ have injective ports labellings, 
        there is two different vertices $w\neq w'\in V(G)$ such that
        $w=\varphi^i(m_1')\neq w'=\varphi^i(m_2')$
      \item Thus, since $\varphi^i$ is an homomorphism, $\varphi^i(w_1w_1')\neq \varphi^i(w_2w_2')$
      \end{itemize}
    \item $m_1=n$ and $m_1'= m_2$
      \begin{itemize}
        \item First, remark that by definition, for every vertex $m\in V(B_{\mapG^i}(n,1))$, there is an edge $nm\in E(\mapG^i)$.
        \item So there are $nm_1', nm_2'\in E(\mapG^i)$, 
        \item From the previous case we know that  $\varphi^i(nm_1')\neq \varphi^i(nm'_2)$.
        \item So $\varphi^i(m_1')\neq \varphi^i(m_2')$ and $\varphi^i(m_1m_1')\neq \varphi^i(m_2m_2')$
      \end{itemize}
    \item $n \neq m_1\neq m_2 \neq n$
      \begin{itemize}
      \item From the previous case $\varphi^i(w_1)\neq \varphi^i(w_2)$ and $\varphi^i(w_1')\neq \varphi^i(w_2')$ in $G$.
      \item Since $\varphi^i$ is an homomorphism, 
        $\varphi^i(w_1w_1')\neq \varphi^i(w_2w_2')$
      \end{itemize}
    \end{enumerate}

    
  \end{proof}

  \begin{lem}\label{lem:surjective-in-map}
    For every $n\in V(\mapG^i\setminus \borderM^{i})$, $\varphi_{|N_{\mapG^i}}(n)$ surjective
  \end{lem}
  \begin{proof}
    \begin{itemize}
    \item Since $\varphi^i$ is an homomorphism,
      every adjacent vertex of $n$ has an image via $\varphi^i$.
    \item Moreover, from Lemma \ref{lem:injective-in-map},  $\varphi^i$ is locally injective from $B_{\mapG^i}(n,1)$ to $B_{G}(\varphi^i(n),1)$.
    \item Consequently, every adjacent vertex of $n$ has a unique image via $\varphi^i$.
    \item Finally, $\Balls[\n]\simeq B_{\mapG^i}(n,1)$ ensures that we map in $\mapG$, every vertex present in $B_{G}(\varphi^i(n),1)$ which is isomorphic by induction to $\Balls[n]$. 
    \item So, $\varphi^i$ is locally surjective for every $n\in V(\mapG^i\setminus \borderM^{i})$ explored phase $i$.

    \end{itemize}
      
\end{proof}

  \begin{lem}\label{lem:injective-on-border}
    for every $n\in V(\borderM^{i})$, $\varphi_{|N_\mapG^i}(n)$ injective
  \end{lem}
  \begin{proof}
    For every phase $i$, for every $n\in \borderM^i$, two case appears:

    If  $n\in V(\borderM^{i-1})$, then $n$ is not explored phase $i$ 
    and by induction, Lemma is proved.

    If $n$ is inserted into $\mapG$ at phase $i$, that is,
    $n\in V( \borderM^i\setminus  \borderM^{i-1})$, then 
    for every  neighbours
    $m\neq m'$ of $n$ in $V(\mapG^{i})$, two cases appear:
    \begin{itemize}
    \item Either, $m$ and $m'$ in $V(\mapG^i\setminus \borderM^{i})$, and in this case,  
      $\Big((m,\delta_m(n)),(m',\delta_{m'}(n))\Big)\in\cloequiv$ .
      By induction,  $\varphi^{i}(m)\neq \varphi^{i}(m')$.
    \item Either, $m$ and $m'$ belong to $\borderM^i$.
      \begin{itemize}
      \item In such a case, since there is always a vertex $\ell$ 
        (resp.  $\ell'$) such that 
        $(\ell,\delta_\ell(n),\delta_\ell(m),\lambda(nm))\in \nivsup$ 
        (resp. $(\ell',\delta_{\ell'}(n),\delta_{\ell'}(m'),\lambda(nm'))\in\nivsup$)
      \item we get that $ [\ell,\delta_{\ell}(m)] \neq [\ell,\delta_{\ell}(n)]=
        [\ell',\delta_{\ell'}(n)]\neq [\ell',\delta_{\ell'}(m')]$.
      \item From previous case,  $\varphi^{i}(\ell)\neq \varphi^{i}(\ell')$.
      \item Since $\mapG^i$ has an injective port labelling, $\delta_n(m)\neq \delta_n(m')$.
      \item So $\varphi^i(m)\neq \varphi^i(m')$ since otherwise, W.l.o.g. 
        $\delta_n(m)\neq \delta_{\varphi(n)}(\varphi(m))$.

    \end{itemize}
  \end{itemize}
  
  
\end{proof}

  We proved the Theorem \ref{thm:inj-surj-map}.

  It remains to prove that $\varphi^i$ preserves triangles in order to prove the correctness of the map construction.
  \begin{lem}
    for every vertex $n\in V(\mapG^{i})$ explored phase $i$, for every triangle $v_1v_2v_3\subseteq B_G(\varphi^i(n),1)$, there is a triangle $m_1m_2m_3\subseteq B_{\mapG^i}(n,1)$ such that $\varphi(m_h)=v_h$ for every $1\leq h \leq 3$ and $\varphi(m_1m_2m_3)=v_1v_2v_3$ 
  \end{lem}

  \begin{proof}

    for every $n\in V(\borderM^{i-1}\setminus\borderM^{i})$ and 
    for every triangle $v_1v_2v_3\subseteq B_G(\varphi^i(n),1)$, 
    \begin{itemize}
    \item From Lemma \ref{lem:injective-in-map}, there is  $m_1m_2m_3\in V(B_{\mapG^{i}}(n,1))$ that $\varphi^i(m_1)\neq \varphi^i(m_2) \neq \varphi^i(m_3)$ and $\varphi(m_1)\varphi(m_2) \neq \varphi(m_2)\varphi(m_3)\neq \varphi(m_3)\varphi(m_1)$
    \item Since we ensure at line \ref{line:checkIso} that $B_{\mapG^i}(n,1)\simeq B_G(\varphi(n),1)$, we get that $m_1m_2m_3$ is a triangle in $\mapG$ if and only if $\varphi^i(m_1)\varphi^i(m_2)\varphi^i(m_3)$ is a triangle in $G$.
    \end{itemize}
  \end{proof}

  Thus, from Lemma above and from Theorem \ref{thm:inj-surj-map}, we get the next corollary, 
  \begin{cor}\label{cor:cover}
    If, at the phase $i$, the agent halts its exploration (without error), then $\mapG^i$ is a simplicial covering of $G$
  \end{cor}
So we finish this proof by the following Theorem,
  \begin{thm}
    If the agent halts its exploration, then the graph is explored
  \end{thm}
  \begin{proof}
    \begin{itemize}
    \item   First, if the agent halts its exploration then it does not find any mistake in its map (Line \ref{line:checkIso}).
    \item Moreover, it explores all vertices in its map.
    \item Suppose the agent halts phase $i$. we get that $\borderM^i=\emptyset$ (all vertices explored) and thus, $\mapG^i$ is locally bijective (injective + surjective).
    \item[$\Rightarrow$] $\varphi$ is a covering
    \item[$\Rightarrow$] (surjective covering) $|\mapG^i|>|G|$
    \item[$\Rightarrow$] $G$ is explored
    \end{itemize}
  \end{proof}
  This Theorem means that the agent cannot halts before exploring every vertex of any graph.
  To conclude this part, we have to prove that the agent always halts on Weetman graphs. 
  \begin{thm}
    For every graph $G\in Weetman$, the algorithm explores $G$ and halts
  \end{thm}
  \begin{proof} First, note that if the agent halts its execution on a Weetman graph $G$,
    from Lemma \ref{prop:weetman-SC} and Lemma \ref{cor:cover}, $\mapG$ 
    is isomorphic to $G$.
    Moreover, since $G$ satisfied the Interval Condition, every vertex $v$
    of $V(G)$ has a unique image $n$ in $V(\mapG)$.
    Since $G$ satisfied the Triangle Condition, every edge $\varphi(n)\varphi(m)$ in $E(G)$
    has a a unique corresponding image $nm$ in $V(\mapG)$.
    Consequently, no error can be found in any execution of $\A$ on a Weetman graph $G$  
    Finally, since the clusters of $G$ form a tree and 
    since the agent performs a DFS over clusters, 
    the agent will reach a phase where all of the nodes in its map are explored.
  \end{proof}

  \begin{rem}
    Let $u,v,w$ be a triple of vertices in $G$ and let $n=\varphi^{i-1}(u), m=\varphi^{i-1}(v)$ and $l=\varphi^{i-1}(w)$ be one pre image of $u,v,w$ in $\mapG$ and $d(v_0,w)>d(v_0,u)=d(v_0,v)$.
    If $u,v,w$ does not respect the Interval condition from $v_0$ in $G$, then the ''top'' vertex $w$ will be duplicated in $\mapG$. In fact, 
    since there is no sequence of adjacent triangles $v_1v_{2}w, ..., v_kv_{k+1}w$ in $G$ explored in the same phase, the agent has no enough pre-vertices relation ($\equiv$) at the end of phase to 
    gather $(n,\delta_u(w)), (m,\delta_u(w))$ in a same equivalence class.
    Figure \ref{img:notIC} illustrates such an error in $\mapG$.
  \end{rem}
  \begin{rem}
    So, for every phase $i$, if two vertices in $\mapG$ are linked to a same vertex in $\borderM^i$, then we know that there corresponding vertex in $G$ are also linked together with a third vertex which is newly discovered.
    From the previous remark, we know that we can duplicate a vertex $w$ of $G$ in $\mapG$.
    But in this case, every 
    vertical edge linking $w$ in $G$ are partitioned and distributed over every ''copies'' of $w$ in $\mapG$ (not duplicated).
  \end{rem}

  \begin{figure}[t]
    \centering
    \includegraphics[height=4cm]{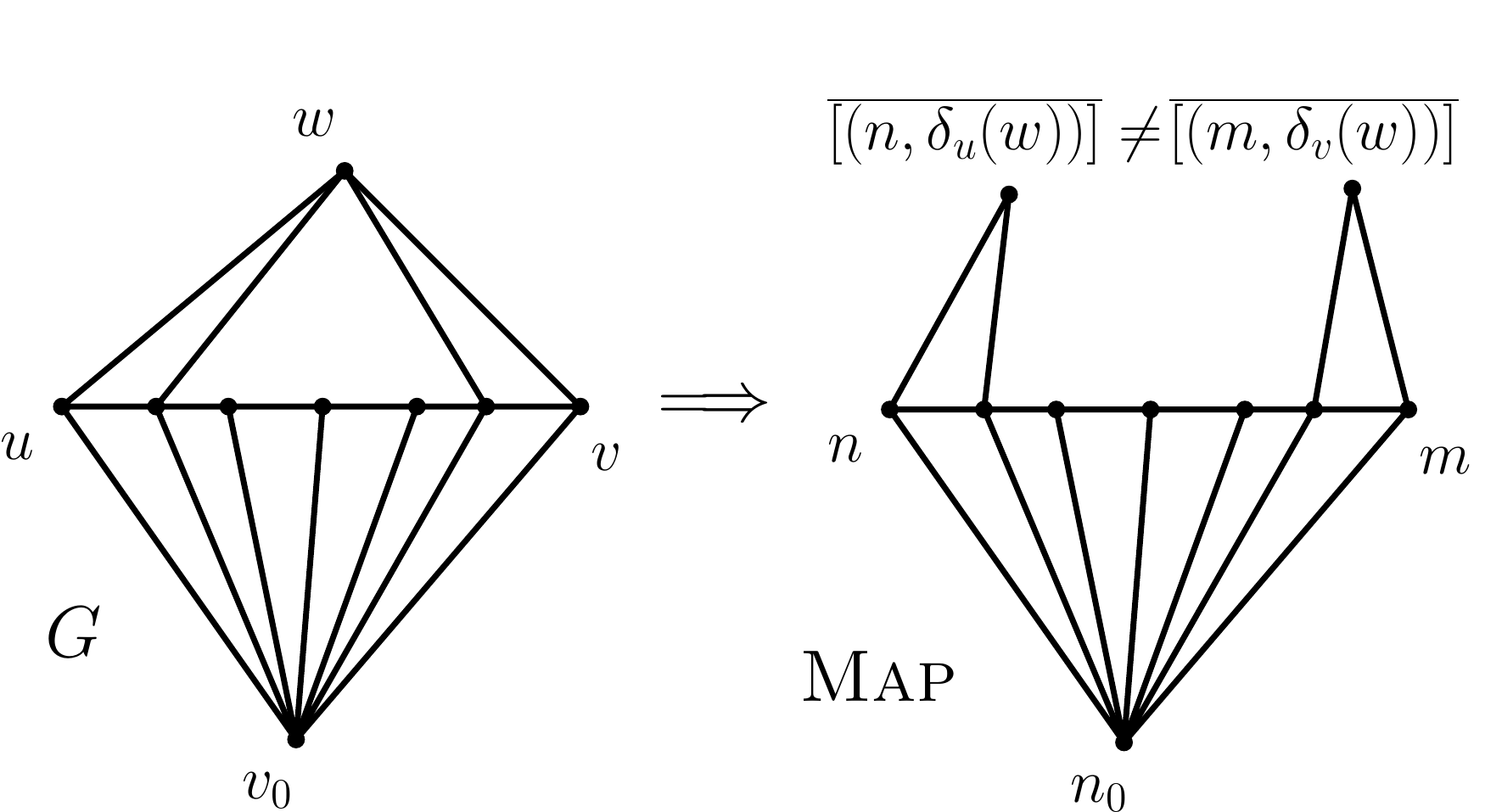}
    \caption{$IC(v_0)$ not satisfied on $u,v,w$}
    \label{img:notIC}
  \end{figure}  
  \subsection{Complexity}
  Let $G$ be a Weetman graph and let $v$ be the homebase of the agent.
  Remark that the total number of clusters $|\CCC(G)|$ is bounded
  by $|V|$, $\sum_{C\in\CCC_G(v)} |V(C)|=|V(G)|$ .

  Note that since in a tree $T$, $E(T)=V(T)-1$, 
  exploring a tree takes at most $2|E(C)|\leq 2|V(C)|$ steps
  even if the agent has to come back to the roots. 
  Moreover, exploring a spanning tree in a graph is a 
  worst case since a cycle in the course of the agent implies 
  that at least one backtracking of the agent is avoided 
  (in front of the spanning tree exploration), which decreases the number of moves.

  Since every edge in a cluster is a horizontal edge, the agent 
  crosses at most $\sum_{C\in \CCC(G)} 2|V(C)|\leq 2|V|$ horizontal edges to explore every cluster.

  Horizontal edges in $C$ can be cross two times more. 
  Once the agent goes to the vertex in $C$ which permits 
  it to reach the next cluster $C'$ to visit.
  Once when the agent has finished the exploration of the
  ''branch'' starting at $C'$ and backtracks to $C$ to go
  to the next cluster $C''$ to visit.
  Note that clusters have to be ordered in a way that
  the agent can go from one to another in a $\mathcal{O}(|V(G)|)$ moves, that is, 
  in a linear number of moves.
  Moreover, a DFS ordering ensures that once a ''branch'' is explored, 
  the agent never returns in this branch. 
  Consequently, we prove that every horizontal edges is crossed by the agent 
  a linear number of times in an execution.

  Since $\mapCCC(G)$ is a tree and the agent performs a DFS on the clusters,
  the agent crosses at most two vertical edges per clusters. 
  We get that the agent crosses a linear number of times 
  vertical edges to go from one to another cluster in an execution. 

  Since the agent explores a spanning tree of $G$, 
  we get that the number of edges crossed is bounded 
  by the number of vertices explores.

  Consequently, since we prove that every edge crossed by the agent a linear number of times, 
  the complexity of Algorithm~\ref{algo:weetman} is achieved in $\mathcal{O}(|V(G)|)$ moves.






  We get our final theorem,
  \begin{thm}
    Algorithm~\ref{algo:weetman} explores and computes a map of
    Weetman graphs with a number of moves that is linear
    in the size of the graph.
  \end{thm}



\section{Conclusion}


We have presented an algorithm that explores all chordal graphs in a
linear number of moves by a mobile agent using binoculars to see the
edges between nodes adjacent to its location. The main contribution is
that the exploration is fast and, contrary to previous works, the agent 
does not need to know the size or the diameter (or bounds)
to halt the exploration.
Using binoculars permits to not visit every edge while still
being assured of having seen all nodes, which is usually not possible
at all without binoculars without additional information about some
graph parameters.

We have actually used properties of chordal graphs that are verified
in a larger class of graphs : the Weetman graphs. This class of graphs
belongs to families of graphs that are defined by local metric
properties and whose clique complexes are simply connected, which is a
necessary condition for linear exploration without knowledge (see \cite{CGN15}). 

Known such families are the family of bridged graphs, or the family of
dismantlable/cop-win graphs that have found numerous application in
distributed computing. 
Chordal and bridged graphs are Weetman and
Algorithm~\ref{algo:weetman} also efficiently explores these graphs.
Dismantlable graphs are not necessarily Weetman, but
given that cop-win graphs can also be defined by an elimination order,
a very interesting open question would be to prove, or disprove, that
there is a linear Exploration algorithm for cop-win graphs.

\footnotesize
\bibliographystyle{alpha}
\bibliography{bib_bino}

\newpage
\appendix
\newtheorem{lemannexe}{Lemma }
\labelformat{lemannexe}{Lemma #1}

\section{Appendix Section \label{sec:weetmanProp}}




From the homotopy relation on cycles, we get the following Proposition.
\begin{prop}\label{notcontractible}
  Given a not contractible cycle $c$ in a graph $G$ and a cycle $c'$
  that is homotopic to $c$, then $c'$ is not contractible.
\end{prop}

\begin{thm}\label{SC-tree-clusters}
  Clusters of a simply connected graph form a tree 
\end{thm}
\begin{proof}
  By contradiction, there is two vertices $x,y\in V(G)$ 
  such that $x,y\in S^k$ and there is no path $p'\subset S^k$ from $x$ to $y$.
  Moreover, to get a cycle of clusters, there
  is a path $q\subset G\setminus B^{k-1}$ from $y$ to $x$.
  Note that, W.l.o.g, there is also a path
  $p:x\to y \subset V(G)$ such that $p\setminus \{xy\}\subset B^{k-1}$.

  We denote by $C(x,y)$ a pair of paths $p:x\to y$ and $q:y\to x$ such that
  $p\setminus \{xy\}\subset B^{k-1}$, $q\subset B^{k+1}$ 
  and $\area(C(x,y))$ is minimal.

   Among every pair of vertices $x',y'$ not relied by a path $p'\subset S^k$,
   let $x,y$ be the couple minimising $\area(C(x,y))$.
   
   Let $a\in p\cap S^{k-1}$ be the vertex such that $ax\in E(p)$ and
   let $b\in q$ such that $bx\in E(q)$. Remark that since 
   $d(v_0,x)=k$, we get that $k\leq d(v_0,b) \leq k+1$

   Since $G$ is simply connected, there is a minimal disk diagram $(D,f)$ for the cycle $C(x,y)=pq$.
   By definition,  $f(\partial D)=C(x,y)$ and we denote by $\tx$ the pre images of $x$ in $D$. W.l.o.g, $f(\tx)=x\in V(G)$.
 
   Since $D$ is a planar triangulation, there is a path $\tts=\tv_1\tv_2...\tv_{\ell-1}\tv_\ell\subset V(D)$ such that for every $1\leq i< \ell$, $f(\tv_{1})=a, f(\tv_{\ell})=b$ and $\tx\tv_i\tv_{i+1}$ is a triangle in $D$.

   Since $d(v_0,a)=d(v_0,x)-1=k-1$ and $k\leq d(v_0,b) \leq k+1$, 
   for every $1\leq i\leq\ell$, $k-1\leq d(v_0,v_i) \leq k+1$.

   Consequently, among every $\tv_j\in \tts$, there is a unique 
   vertex $\tv_i\in \tts$ such that $d(v_0,v_i)=k$ and 
   for every $1\leq h< i$, $d(v_0,v_h)=k-1$.

   Remark that there is a couple of 
   paths $p':v_i\to y$ and  $q':y\to v_i$ such that
   $p'=p\setminus\{ax\}\cup\{av_2...v_i\}$, 
   $q'=q\cup \{xv_i\}$.
   Moreover, since $x$ and $y$ are not relied by a path in $S^k$ and 
   since $v_ix\in E(S^k)$, $v_i$ and $y$ are also not relied
   by a path in $S^k$.

   Since $p'\setminus \{v_iy\} \subset B^{k-1}$, and $q'\subset G\setminus B^{k-1}$, it remains to prove that $\area(p'q')$ is smaller than $\area(pq)$.

   Since $D'=D\setminus \{uv_h,\forall 1\leq h < i\}$ is a disk diagram for $p'q'$, it is easy to see that $\area(D')<\area(D)$ since for every $1\leq h < i$, the triangle $\tu\tv_h\tv_{h+1}$ does not appear in $D'$.

   We get a contradiciton on the choice of $x,y$.
   Consequently, there is no couple of vertices inside $S^k$, for every $k$, which are not relied by a path inside $S^k$. 
   We prove that there is no cycle of clusters in simply connected graph and thus, clusters form a tree.
\begin{figure}[t]
  \centering
  \includegraphics[height=3cm]{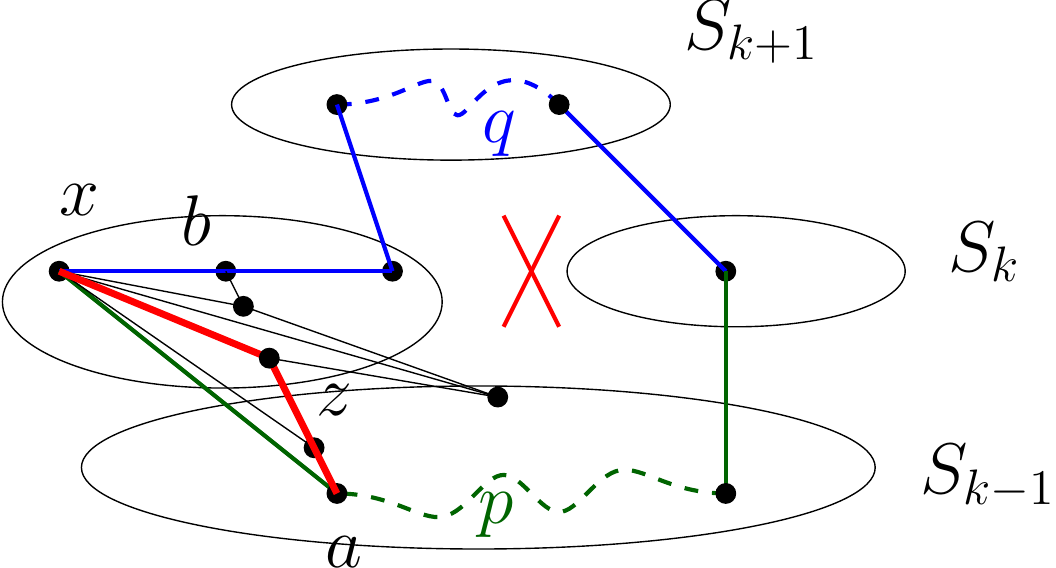}
\end{figure}

\end{proof}

\end{document}